\documentclass{article}

\usepackage{arxiv}

\usepackage[utf8]{inputenc} 
\usepackage[T1]{fontenc}    
\usepackage{hyperref}       
\usepackage{url}            
\usepackage{booktabs}       
\usepackage{amsfonts}       
\usepackage{nicefrac}       
\usepackage{microtype}      
\usepackage{lipsum}		
\usepackage{graphicx}
\usepackage{doi}
\usepackage{amsmath,amsthm}
\usepackage{algorithm}
\usepackage{algorithmic}
\usepackage{multirow}
\usepackage{tabularx}
\usepackage{hyperref}

\newtheorem{definition}{Definition}
\newtheorem{theorem}{Theorem}
\newtheorem{proposition}{Proposition}
\newtheorem{corollary}{Corollary}

\renewcommand{\det}{\operatorname{det}}
\newcommand{\Spec}{\operatorname{Spec}}
\renewcommand{\deg}{\operatorname{deg}}
\newcommand{\X}{\operatorname{\mathcal{X}}}
\newcommand{\C}{\operatorname{\mathcal{C}}}

\newcommand{\A}{\operatorname{\hat{A}}}
\newcommand{\Al}{\operatorname{\hat{A}^l}}
\newcommand{\Au}{\operatorname{\hat{A}^u}}
\newcommand{\As}{\operatorname{\hat{A}^s}}
\newcommand{\Aone}{\operatorname{\hat{A}_{I}}}
\newcommand{\Atwo}{\operatorname{\hat{A}_{II}}}
\newcommand{\Pone}{\operatorname{\hat{P}_{I}}}
\newcommand{\Ptwo}{\operatorname{\hat{P}_{II}}}
\newcommand{\G}{\operatorname{\hat{\mathcal{G}}}}

\renewcommand{\hat}{\widehat}

\title{Link Partitioning on Simplicial Complexes \\Using Higher-Order Laplacians}

\author{ {\hspace{1mm}Xinyi Wu}
	\\MIT\\
	\texttt{xinyiwu@mit.edu} \\
	\And
	{\hspace{1mm}Arnab Sarker} 
	\\MIT\\
	\texttt{arnabs@mit.edu} \\
    \And
	{\hspace{1mm}Ali Jadbabaie} \\
	MIT\\
	\texttt{jadbabai@mit.edu} \\
}

\date{}

\renewcommand{\shorttitle}{Link Partitioning on Simplicial Complexes Using Higher-Order Laplacians}

\hypersetup{
pdftitle={A template for the arxiv style},
pdfsubject={q-bio.NC, q-bio.QM},
pdfauthor={David S.~Hippocampus, Elias D.~Striatum},
pdfkeywords={First keyword, Second keyword, More},
}

\begin{document}
\maketitle

\begin{abstract}
	Link partitioning is a popular approach in network science used for discovering overlapping communities by identifying clusters of strongly connected links. Current link partitioning methods are specifically designed for networks modelled by graphs representing pairwise relationships. Therefore, these methods are incapable of utilizing higher-order information about group interactions in network data which is increasingly available. Simplicial complexes extend the dyadic model of graphs and can model polyadic relationships which are ubiquitous and crucial in many complex social and technological systems. In this paper, we introduce a link partitioning method that leverages higher-order (i.e. triadic and higher) information in simplicial complexes for better community detection. Our method utilizes a novel random walk on links of simplicial complexes defined by the higher-order Laplacian—a generalization of the graph Laplacian that incorporates polyadic relationships of the network. We transform this random walk into a graph-based random walk on a lifted line graph—a dual graph  in which links are nodes while nodes and higher-order connections are links—and optimize for the standard notion of modularity. We show that our method is guaranteed to provide interpretable link partitioning results under mild assumptions. We also offer new theoretical results on the spectral properties of simplicial complexes by studying the spectrum of the link random walk. Experiment results on real-world community detection tasks show that our higher-order approach significantly outperforms existing graph-based link partitioning methods.
\end{abstract}


\section{Introduction}
Community detection is a central research topic in network science and provides major insights into the structure and function of complex networks. Although the precise definition of a community largely depends on the application context~\cite{Schaub2017TheMF}, community detection has almost always been regarded as a problem of determining groups of \emph{nodes} who share dense connections~\cite{FORTUNATO201075}. However, nodes often belong to multiple communities in networks, thus complicating the task of community assignments~\cite{Ahn2010LinkCR, Evans2009LineGL,Palla2005UncoveringTO}. For example, in a social network where each node represents a person, a person could be family to someone, a friend to another person, and a co-worker to someone else. 

Many overlapping community detection methods have been proposed to accommodate this problem~\cite{Ahn2010LinkCR,Evans2009LineGL,Evans2010LineGO,Palla2005UncoveringTO,Coscia2012DEMONAL,Gopalan2013EfficientDO,Lancichinetti2009DetectingTO}. 
In this paper, we focus on the “link community” paradigm, where communities are redefined as sets of closely interrelated \emph{links} (edges) rather than nodes~\cite{Ahn2010LinkCR,Evans2009LineGL,Evans2010LineGO}. 
Unlike nodes, links in networks usually form for one dominant reason. 
In the above example, links form between two people either because they come from the same family, share common interests, or work together. This new definition of a community allows us to naturally find overlapping structures by redefining a node's set of communities as incident link communities. Compared with node-based approaches, link partitioning has been found to reveal community structures with better quality in networks from various domains~\cite{Ahn2010LinkCR}. 

Current link partitioning algorithms are exclusively designed for graph representations of networks \cite{Ahn2010LinkCR,Evans2009LineGL,Evans2010LineGO,Shi2013ALC,Orgaz2018AMG,Deng2017FindingOC,Lee2017InverseRL,Zhang2015SymmetricNM}, which only model pairwise relationships between nodes. 
However, higher-order relationships and interactions in groups of more than two entities, which cannot be modelled by a graph, are ubiquitous and essential to understanding the structures and the behavior of complex systems in many fields\cite{Medaglia2015CognitiveNN,Milo2002NetworkMS,Holland1970AMF,Granovetter1973TheSO} (Figure~\ref{fig:intro}). The complexity of real-world networks necessitates a higher-order informed link partitioning method that fundamentally captures polyadic interactions and relationships among nodes.

\begin{figure}[t]
    \centering
    \includegraphics[width = 10cm]{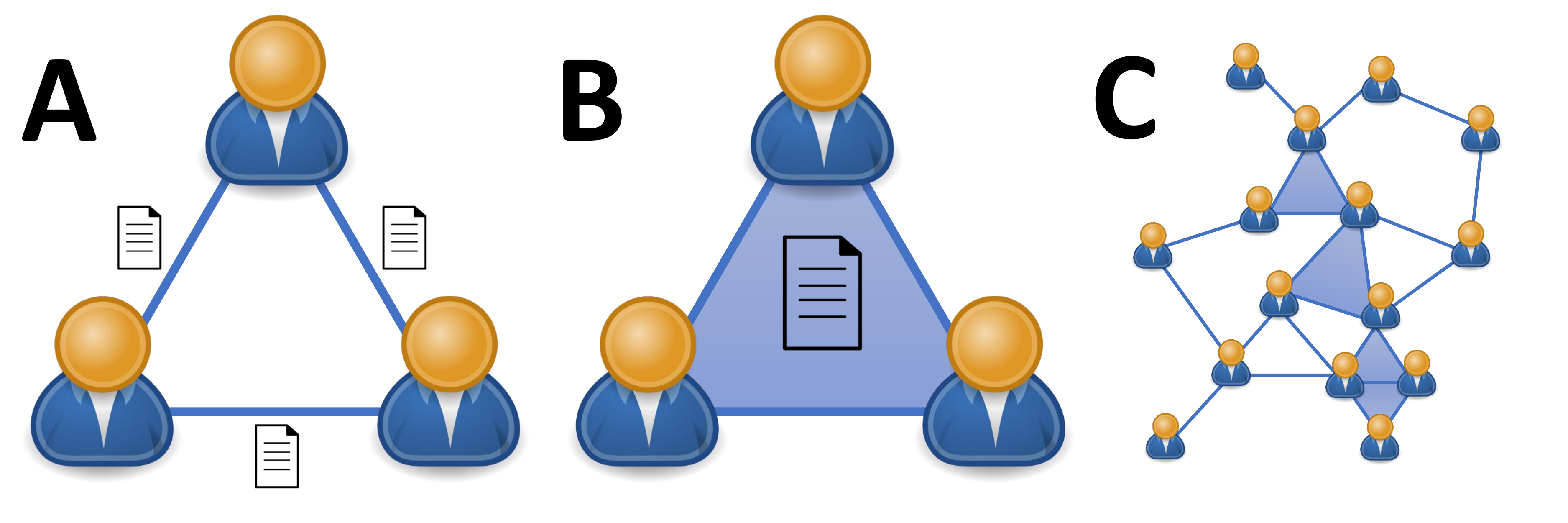}
    \caption{Example of the distinction between pairwise and polyadic relationship. \textbf{A}: Each pair of the three authors has coauthored a (potentially different) paper (three pairwise relationships). \textbf{B}: The three authors have coauthored a paper together (one polyadic relationship). Relationship A is different from B, and graphs are insufficient to distinguish A from B. \textbf{C}: Our method uses simplicial complexes to model both pairwise and polyadic relationships and provides substantially better overlapping community detection results.}
    \label{fig:intro}
\end{figure}

In this paper, we propose a link partitioning method that accounts for higher-order group behavior in networks. Our method models networks with higher-order (i.e. beyond dyadic) information using simplicial complexes and leverages tools from algebraic topology and combinatorial Hodge theory~\cite{Lim2020HodgeLO,Hatcher2002AlgebraicT}. 
In particular, our method derives link communities from link random walks based on higher-order Hodge Laplacians. This can be seen as a higher-order analogue of the use of spectral graph theory for node community detection, where one can interpret the action of normalized graph Laplacians as random walks on graphs.
 Motivated by the importance of triangular structures in social networks \cite{Holland1970AMF,Granovetter1973TheSO,Newman2001ClusteringAP}  and their success in finding community structures in practice \cite{Benson2016HigherorderOO,Yin2017LocalHG,Sotiropoulos2021TriangleawareSS}, we focus on triadic connections encoded by filled triangles, \textit{i.e.} simplices of dimension 2, in simplicial complexes.

Inspired by a novel higher-order random walk recently introduced in Schaub et al.~\cite{Schaub2020RandomWO}, we utilize the notion of a link-based random walk on oriented simplicial complexes to propose our new method for link partitioning. The random walk consists of what we denote as an \textit{upper walk}, a \textit{lower walk}, and a set of self-loops. The proposed link-based random walk can also be transformed into a weighted graph random walk on a \textit{lifted line graph} $\G$ where all oriented links become nodes, while nodes and filled triangles become links (Figure~\ref{fig:lifting}). We take advantage of this transformation and use the transformed walk in conjunction with the Louvain method~\cite{Blondel2008FastUO} to design a modularity optimization algorithm for link partitioning. Since two orientations of each link are present in the lifted line graph $\G$, we show that under some mild conditions that generally hold for real networks, the proposed algorithm will always cluster the two orientations of each link into the same community. That is, the community partition in the lifted line graph $\G$ has a direct interpretation as link communities in the original higher-order network.

In addition to the algorithm itself, we connect to the spectral theory for simplicial complexes by
proving a set of fundamental spectral properties for the lifted line graph random walk $\hat{\textbf{P}}$. We identify the equivalence between the link random walk and a diffusion process on the simplicial complex propagating via the higher-order Laplacian~\cite{Parzanchevski2017SimplicialCS,Mukherjee2016RandomWO}. Hence the lifted line graph random walk $\hat{\textbf{P}}$ is directly related to the random process and their spectral properties are closely connected.
Our results show that the spectrum of the random walk matrix $\hat{\textbf{P}}$ consists of two parts, which we call the ``even'' part and the ``odd'' part. The even part has symmetry in the space on the lifted line graph $\G$, whereas the odd part corresponds to the spectrum of the Laplacian and thus it is connected to the homology of the simplicial complex. There have been constant developments over years trying to understand the spectral properties and building a spectral theory for simplicial complexes \cite{Parzanchevski2017SimplicialCS,Mukherjee2016RandomWO,Parzanchevski2016IsoperimetricII,Horak2011SpectraOC,Kaufman2018HighOR}. Our work provides a new venue for understanding the spectral properties of simplicial complexes through the lifted line graph $\G$ and the random walk $\hat{\textbf{P}}$.

Finally, we conduct numerical experiments on a number of community detection tasks. We first show that using higher-order information improves the identification of closely interrelated links in a synthetic network. We then test our proposed method in seven real-world networks where higher-order information is available. On average, our method has $12\%$ improvement over the graph baselines. The comparisons suggest that higher-order information is valuable and thus our higher-order framework overall provides better solutions.

In summary, our paper develops a link partitioning method that incorporates higher-order information in networks.
\begin{itemize}
    \item Theoretically, the proposed method generalizes the use of graph Laplacian for community detection to higher-order network data structure. It provides new insights into the spectral theory for simplicial complexes.
    \item Algorithmically, the method takes advantage of the lifted line graph transformation and the Louvain method for community detection. We show that optimizing modularity in the lifted line graph using the Louvain method is guaranteed to produce interpretable link partitioning results under realistic conditions.
    \item Empirically, we find that our higher-order method results in substantial improvements in discovering overlapping community structure. 
\end{itemize}
Our work uses simplicial complexes to exploit rich and valuable higher-order network data and opens up new possibilities for higher-order network problems that apply across a broad set of domains.

\section{Related Work}
\paragraph{\textbf{Link-based Community Detection}}
The problem of overlapping community detection using links in networks was first introduced in two seminal works \cite{Ahn2010LinkCR,Evans2009LineGL}.  
Evans and Lambiotte \cite{Evans2009LineGL} consider line graphs, a transformation of a network where links become nodes and nodes become links, and optimize weighted modularity scores of the line graph to find communities of links. 
Ahn et al. \cite{Ahn2010LinkCR} propose a pairwise similarity metric between links and builds a hierarchical dendrogram to determine link communities. The paper also proposes an objective function called partition density and finds communities at its optimal value in the dendrogram. 
Since then, several methods have been developed as improvements of these two methods \cite{Evans2010LineGO,Shi2013ALC,Orgaz2018AMG,Deng2017FindingOC,Lee2017InverseRL,Zhang2015SymmetricNM}. Nonetheless, although higher-order network analysis has become a fundamental research topic in network science~\cite{Benson2016HigherorderOO,Benson2018SimplicialCA,Benson2021HigherorderNA,Milo2002NetworkMS}, to our best knowledge, there has not been any work done incorporating higher-order information about simultaneous interactions in groups into link partitioning.

\paragraph{\textbf{Higher-Order Network Analysis}}
Traditional network science studies complex networks using graph representations, which model pairwise interactions between entities via links.  Yet many real-world complex systems involve simultaneous relationships of more than two entities. Network scientists thus have to go beyond graph-based models in order to study the higher-order structures and dynamics of complex systems~\cite{Benson2021HigherorderNA}. 
Higher-order network analysis uses the ideas of motifs~\cite{Benson2016HigherorderOO,Yin2017LocalHG}, hypergraphs~\cite{Chodrow2021GenerativeHC}, simplicial complexes~\cite{Benson2018SimplicialCA,Schaub2020RandomWO}, multilinear and tensor algebra~\cite{Benson2015TensorSC}, \textit{etc.}, to model complex systems with higher-order information. 
Among these models, simplicial complexes can be seen as a generalization of graphs from the point of view of algebraic topology and thus allow researchers to use rich mathematical tools. Specifically, the higher-order Hodge Laplacian can be seen as a generalization of the powerful graph Laplacian.  Recent works have used higher-order Laplacians for trajectory prediction~\cite{Glaze2021PrincipledSN}, signal processing~\cite{Barbarossa2020TopologicalSP}, identifying tie strength~\cite{Sarker2021HigherOI}, and label propagation~\cite{Mukherjee2016RandomWO}. 
 In addition, Ebli and Spreemann~\cite{Ebli2019ANO} use the higher-order Laplacian to cluster data supported on links of simplicial complexes. However, this specific method is not appropriate for community detection tasks as it (i). does not cluster links that are closely interrelated in the way communities are usually defined; (ii). often leaves a large part of the links in networks unclustered. 

 \section{Preliminaries}
 In this section, we review the background material. We first introduce the standard notion of modularity from a random walk perspective in Section~\ref{MODULARITY}. This motivates our method of using a random walk on links to find good link partitions in terms of modularity. We then introduce the basic concepts of simplicial complexes and higher-order Hodge Laplacians in Section~\ref{SC} and Section~\ref{HL}, respectively. They build the foundation of our link partitioning method on higher-order networks.
\vspace{-1ex}

\subsection{Modularity from a Random Walk Perspective}\label{MODULARITY}
To motivate our method of using a random walk on links for link partitioning, let us first consider the well-known concept of \emph{modularity}~\cite{Newman2006ModularityAC} in terms of random walks on nodes.
Given an undirected (weighted) graph $\mathcal{G}$ with adjacency matrix $A$,  the modularity $Q$  used for evaluating quality of a node partition $\mathcal{P}$, is defined as 
\begin{equation}
    Q = \frac{1}{2m}\sum_{C\in \mathcal{P}}\sum_{i,j\in C}[A_{ij}-\frac{k_ik_j}{2m}]\,, 
\label{eq:mod}
\end{equation}
where $k_i = \sum_{j}A_{ij}$ is the sum of the weights of the links attached to node $i$, $2m = \sum_{ij}A_{i,j}$ is the sum of all of the link weights in the graph, and $C$ runs over all the communities in $\mathcal{P}$. 

A random walk interpretation of the modularity $Q$ is the follows \cite{Delvenne12755,Lambiotte2008LaplacianDA}: denote the probability of a random walker on node $i$ at time step $t$ as $p_{i,t}$, where the dynamics are given by the standard unbiased random walk $p_{i,t+1} = \sum_{j}\frac{A_{ij}}{k_j}p_{j,t}$. When the given network is undirected, connected and non-bipartite, one can show that the stationary solution of the dynamics is $p_i^\star = k_i/2m$. 
Consider a community $C\in \mathcal{P}$. If the system is at equilibrium, the probability that a random walker stays in $C$ during two successive time steps is $\sum_{i,j\in C}\frac{A_{ij}}{k_j}\frac{k_j}{2m}$, whereas the probability of finding two independent walkers in $C$ is $\sum_{i,j\in C}\frac{k_ik_j}{(2m)^2}$. Hence one can reinterpret $Q$ as a summation over the communities of the difference of the two probabilities. This interpretation suggests a natural generalization of modularity that allows one to tune its resolution as follows: $Q$ is based on paths of length one,  but one can generalize it to paths of arbitrary length $t$ as
\begin{equation}
    R(A,t) = \frac{1}{2m}\sum_{C\in\mathcal{P}}\sum_{i,j\in C}[(T^t)_{ij}k_j - \frac{k_ik_j}{2m}]\,,
\label{eq:stability}
\end{equation}where $T_{ij} = A_{ij}/k_j$. $R(A,t)$ is called the \emph{stability} of the partition~\cite{Delvenne12755} and optimizing the stability typically leads to partitions made of larger and larger communities for increasing $t$~\cite{Delvenne12755,Lambiotte2008LaplacianDA}.

The above random walk formulation of modularity suggests that one should look at a random walk process traversing the links of a network to find communities of links. Yet previous link random walks are solely designed for graphs~\cite{Evans2009LineGL,Evans2010LineGO}. In order to define a link random walk on a higher-order network, we introduce simplicial complexes and higher-order Hodge Laplacians for higher-order network modelling.

\subsection{Simplicial Complexes}\label{SC}
Let $V$ be a finite set of nodes. 
A \emph{k-simplex} $\mathcal{S}^k$, often referred to as a simplex of dimension $k$, is a subset of $V$ with $k+1$ elements. 
A \emph{simplicial complex} (SC) $\mathcal{X}$ is a set of simplices with the property that if $\mathcal{S} \in  \mathcal{X}$, then all subsets of $\mathcal{S}$ are also in $\mathcal{X}$. 
We use $\X^k$ to denote the subset of all $k$-simplices in $\X$. 
A $k$-simplex has $k+1$ subsets of dimension $k-1$, which are called \emph{faces}. 
If $\mathcal{S}^{k-1}$ is a face of simplex $\mathcal{S}^k$, 
$\mathcal{S}^k$ is called a $\emph{co-face}$ of $\mathcal{S}^{k-1}$~\cite{Hatcher2002AlgebraicT}. 
While these definitions are combinatorial, geometrically one can think of 0-simplices as nodes, 1-simplices as links (edges), 2-simplices as filled triangles, 3-simplices as tetrahedra, and so forth. A graph can thus be interpreted as an SC where all simplices have dimension at most $1$.

Two $k$-simplices in an SC $\mathcal{X}$ are \emph{upper adjacent} if they are both faces of the same $(k+1)$-simplex and are \emph{lower adjacent} if both share a common face. 
For any simplex in $\X$, we define its degree, denoted by $\deg(\cdot)$, to be the number of its co-faces.  
  
We endow arbitrarily each simplex with an \emph{orientation}. An orientation can be seen as a chosen ordering of the elements of a simplex, modulo even permutations. That is, for a $k$-simplex $\mathcal{S}^k = \{i_0,i_1,...,i_k\} \in \X^k$, an orientation of $\mathcal{S}^k$ would be $[i_0,i_1,...,i_k]$. Performing an even permutation ($e.g.\,[i_1,i_2,i_0,...,i_k]$) leads to an equivalent orientation. For simplicity, we choose the reference orientation of the simplices induced by the ordering of the node labels $\{s_i^k = [i_0,...,i_k],i_0<...<i_k\}$. 
\vspace{-1ex}
\subsection{Boundary Operators and Hodge Laplacians}\label{HL}
For each dimension $k$, we define the finite-dimensional vector space $\C_k$ with coefficients in $\mathbb{R}$, whose basis elements are the oriented simplices $s_i^k$. An element $c_k \in \C_k$ is called a \emph{k-chain}, and is a linear combination of the basis elements $c_k = \sum_{i}a_is_i^k$. Thus, each element $c_k \in \C_k$ can be represented by a vector $a = [a_1,...,a_{n_k}]$, where $n_k = |\X^k|$ is the number of $k$-simplices in $\X$. A change of the orientation of the basis $s_i^k$ is defined to be a flip of the sign of its coefficient $a_i$. 

Given the space of chains $\C_k$, we define the linear boundary maps $\partial_k: \C_k \to \C_{k-1}$~\cite{Hatcher2002AlgebraicT} as
\vspace{-1ex}
\begin{equation*}
    \partial_k([i_0,...,i_k]) = \sum_{j=0}^k(-1)^j[i_0,...,i_{j-1},i_{j+1},...,i_k]\,.
\vspace{-1ex}
\end{equation*}
$\partial_k$ maps any $k$-chain to a sum of its faces, with the chosen orientations. Once fixing a basis, these operators can be represented by a matrix. We denote the matrix representation of $\partial_k$ by $B_k$, where each column has exactly $k+1$ nonzero entries. Meanwhile, for each $\partial_k$, there exists a co-boundary map $\partial_k^\top: \C_k \to \C_{k+1}$, which can be represented by the adjoint of $B_k$, denoted $B_k^\top$. Based on the sequence of boundary and co-boundary operators, the $k^{th}$ \emph{Hodge Laplacian}~\cite{Lim2020HodgeLO}, also known as the Eckmann Laplacian~\cite{Eckmann1944HarmonischeFU}, is defined as 
\begin{equation*}
    L_k = B_k^\top B_k + B_{k+1}B_{k+1}^\top\,.
\end{equation*}
\vspace{-0.3ex}We note that the standard combinatorial graph Laplacian is a special case of the Hodge Laplacian, as it corresponds to $L_0 = B_1B_1^\top$, since $B_0 = 0$. The matrix $L_1$, which is called the Hodge 1-Laplacian, is the primary focus of this paper.

\section{Methodology}\label{METHOD}
We now develop our higher-order link partitioning methodology. In Section~\ref{stochastic_lifting}, we introduce the normalized Hodge 1-Laplacian $\mathcal{L}_1$ which gives rise to a higher-order informed random walk $\hat{\textbf{P}}$ on links of a simplicial complex $\X$. We show that this random walk can be seen as a graph random walk on a lifted line graph transformation $\G$ of $\X$. Then in Section~\ref{Louvain}, we find the optimal link partition through optimizing the modularity in $\G$ using the Louvain method. We prove that our method is guaranteed to produce interpretable link partitioning results under mild conditions. Finally, in Section~\ref{Phat}, we derive several spectral properties about $\hat{\textbf{P}}$ and connect to the spectral theory for simplicial complexes.

Throughout the rest of the paper, we assume that the graph skeleton of the SC $\mathcal{X}$ (i.e. we ignore all the 2-simplices) is undirected, connected and non-bipartite.
\vspace{-1ex}
\subsection{Stochastic Lifting of the Normalized Hodge 1-Laplacian}\label{stochastic_lifting}

Inspired by~\cite{Schaub2020RandomWO}, given an SC modelled network $\mathcal{X}$, we define the normalized Hodge 1-Laplacian $\mathcal{L}_1$ as follows.

\begin{definition}
Consider a simplicial complex $\mathcal{X}$ up to dimension 2 with boundary operators $B_1$ and $B_2$. The normalized Hodge Laplacian matrix is defined as
\begin{equation*}
    \mathcal{L}_1 = [B_1^\top B_1 + B_2B_2^\top]D_{tot}^{-1}\,,
\end{equation*}
where $D_{tot}$ is a diagonal matrix with diagonal entries $[D_{tot}]_{[i,j],[i,j]} = deg(i) + deg(j) + 3deg(i,j)$.
\end{definition} 
A way to understand the action of the Hodge 1-Laplacian on any vector $f \in \mathbb{R}^{n_1}$ representing a flow on links is to consider a higher-dimensional, lifted state space, where both possible orientations for each link is present~\cite{Schaub2020RandomWO}. Specifically, we can utilize a natural inclusion map $V$ which maps any link-flow by explicitly representing both orientations. We choose appropriate bases such that the matrix representation of $V$ is 
\begin{equation*}
    V = \begin{bmatrix}
+I_{n_1}\\
-I_{n_1}
\end{bmatrix} \in \mathbb{R}^{2n_1\times n_1}\,,
\end{equation*}
where $I_{n_1}$ is the identity matrix of dimension $n_1$, the number of links.
Given a link-flow $f$, the lifted link-flow is $\hat{f} = Vf = (f^\top,-f^\top)^\top$. Throughout the rest of the paper, we use $\hat{\cdot}$ to indicate objects that are related to the lifted space. 

Similarly, a lifting of a matrix using $V$ is defined as:

\begin{definition}
We call a matrix $N \in \mathbb{R}^{2n_1\times 2n_1}$ a lifting of a matrix $M \in \mathbb{R}^{n_1\times n_1}$ if
\begin{equation}
    V^\top N = MV^\top\,.
\label{def:lifting}
\end{equation}
\end{definition}
This definition implies that if $M$ has a lifting $N$, then $M = \frac{1}{2}V^\top NV$, as $V^\top V = 2I_{n_1}$. As a result, the action of $M$ can be interpreted in terms of lifting $V$, then a linear transformation represented by $N$, and finally a projection $\frac{1}{2} V^\top$ into the original lower dimensional space. In this context, ``projection'' just refers to a mapping into a lower-dimensional space. 

We remark that this ``lift-propagate-project'' procedure with $M$ equals the Hodge 1-Laplacian $L_1$ is in fact equivalent to a higher-order diffusion dynamic called the ``expectation process'' proposed in~\cite{Parzanchevski2017SimplicialCS}. Under appropriate normalization, the process will converge to the kernel of the Laplacian and thus is related to the homology of $\X$~\cite{Parzanchevski2017SimplicialCS,Mukherjee2016RandomWO}. 
 
 We now show that the lifting of our normalized Hodge Laplacian $\mathcal{L}_1$ corresponds to a random walk on a lifted line graph transformation  $\G$ of the original simplicial complex $\X$. This is a higher-order generalization of the relationship between the normalized graph Laplacian and the random walk on graph in spectral graph theory~\cite{Chung1996SpectralGT}.  To state our results compactly, we define the following matrices:
 $\hat{B}_1 := B_1V^\top = \begin{bmatrix}
 B_1 & -B_1
 \end{bmatrix}, \hat{B}_2 := VB_2 = \begin{bmatrix}
 B_2 & -B_2
 \end{bmatrix}^\top$.
We denote $\hat{B}_i^+$ and $\hat{B}_i^-$ as the positive part and the negative part of these matrices, respectively, i.e. $(\hat{B}_i^+)_{ab} = \max\{(\hat{B}_i)_{ab}, 0\}, (\hat{B}_i^-)_{ab} = \max\{(-\hat{B}_i)_{ab}, 0\}$. 
\begin{theorem}
The matrix $(I - \mathcal{L}_1)/2$ has a lifting, i.e. there exists a stochastic matrix $\hat{\textbf{P}}$ such that
$$(I - \mathcal{L}_1)V^\top = 2V^\top \hat{\textbf{P}}\,,$$
where $\hat{\textbf{P}}$ corresponds to a random walk on an undirected lifted line graph $\G$ with adjacency matrix $\hat{A} = \hat{A}^l + \hat{A}^u + \hat{A}^s$, where
\begin{itemize}
    \item $\hat{A}^l = [\hat{B}^-_1]^\top\hat{B}^+_1+[\hat{B}^+_1]^\top\hat{B}^-_1$;
    \item $\hat{A}^u = \hat{B}^+_2[\hat{B}^-_2]^\top+\hat{B}^-_2[\hat{B}^+_2]^\top$;
    \item $\hat{A}^s$ is a diagonal matrix with $[\hat{A}^s]_{[i,j],[i,j]} = deg(i) + deg(j) + 3deg(i,j)$.
\end{itemize}
\label{thm:lifting}
\end{theorem}

 The proof is in Appendix~\ref{app:pf}.  In the lifted line graph $\G$, nodes are oriented links of $\X$, while links are nodes
and filled triangles of $\X$. Figure~\ref{fig:lifting} gives an illustration of the lifting process, the lifted line graph $\G$, and the three components of $\A$. One can interpret $\Al$ as describing the connections between lower adjacent links that are aligned with respect to the reference orientations, while $\Au$ as describing the connections between upper adjacent links that are not aligned respect to the reference orientations. Finally, by construction, the random walk $\hat{\textbf{P}}$ contains self-loops encoded by $\As$ and thus could be seen as a lazy random walk. If we denote $d_{[i,j]} = 2(deg(i) + deg(j) + 3deg(i,j))$, $\hat{\textbf{P}}$ has an interpretation in words as follows: starting from state $[i,j]$, with probability $1/2$, the walker remains in the same state; with probability $(deg(i)+deg(j))/d_{[i,j]}$, the walker transitions to a lower adjacent state (\textit{lower walk}); with probability $3deg(i,j)/d_{[i,j]}$, the walker transitions to an upper adjacent state (\textit{upper walk}).

Moreover, the adjacency matrix $\A$ of $\G$ has a special block form: 
\begin{proposition}
The adjacency matrix $\hat{A}$ of $\hat{\mathcal{G}}$ is symmetric and has the following block form:
\begin{equation*}
    \begin{bmatrix}
\hat{A}_I & \hat{A}_{II} \\
\hat{A}_{II} & \hat{A}_{I}
\end{bmatrix}
\end{equation*}
\label{prop:block}
\end{proposition}
\vspace{-2ex}
In words, $\Aone$ describes the connectivity between each pair of oriented links either both with the reference orientations or both with the reverse orientations. $\Atwo$ describes the connectivity between each pair of oriented links such that one has the reference orientation and the other has the reverse orientation.

\begin{figure*}
    \centering
    \vskip 0pt
    \includegraphics[width =\linewidth]{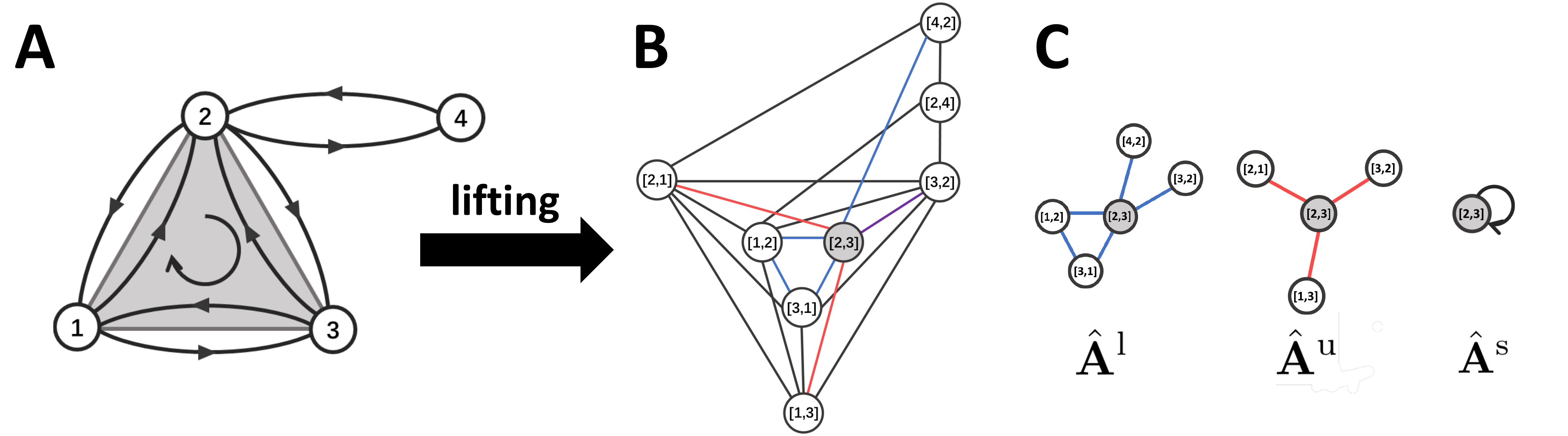}
    \caption{Lifting of a simplicial complex with simplices up to dimension $2$. \textbf{A}: a simplicial complex where each link is represented in both possible orientations. The shaded area represents a 2-simplex. \textbf{B}: The lifted line graph $\G$ of the simplcial complex (self-loops are omitted for clarity). We can interpret each link $(a,b)$ in the original complex as giving rise to two states $[a,b]$ and $[b,a]$ on $\G$ with $2n_1$ nodes. \textbf{C}: $\G$ has three types of connections, represented by $\Al$, $\Au$ and $\As$, respectively. Here we show the egoview of the grey node $[2,3]$ in terms of the three types of connections.}
    \label{fig:lifting}
\end{figure*}

\subsection{Modularity Optimization Using the Louvain Method}\label{Louvain}
Having defined $\hat{\textbf{P}}$ and $\A$, we plug $\A$ into the definition of modularity~(\ref{eq:mod}) and (\ref{eq:stability}) to find the optimal link partitions in terms of modularity in $\G$ with varying resolution. In practice, such optimal partitions can be found by standard modularity optimization algorithms. One of the most popular and state-of-the-art modularity optimization method is the Louvain method~\cite{Blondel2008FastUO}. Since we include two orientations of each link as nodes in the lifted line graph $\hat{\mathcal{G}}$, there is a possibility that applying a modularity optimization method directly to $\hat{\mathcal{G}}$ could break the two orientations of a same link into different groups. Such a partitioning result would be hard to interpret. However, we show that such a problem generally does not exist for the Louvain method. 

The Louvain algorithm has two phases: partitioning and reconstructing. In the first phase, one starts by putting each node in the graph into a distinct
community. For each node $i$, the algorithm performs two
calculations: (i). Compute the modularity change $\Delta Q$ when putting node $i$ into the community of some neighbor $j$; (ii). Move $i$ to the community of node $j$ that yields the
largest gain in $\Delta Q$. This process repeats until step 2 can no longer improve the modularity, i.e. when a local maximum of the modularity is attained. Then we enter the second phase, where we collapse the existing communities to single nodes and build a new network where nodes are the communities from the previous phase, and links weights are aggregated accordingly. Once the new network is created, the second phase ends and the first phase can be re-applied to the new network.

We now formally state the previous claim that the Louvain method applied to $\hat{G}$ will always group the two orientations of each link into the same community:
\begin{theorem}
Given the lifted line graph $\hat{\mathcal{G}}$ with adjacency matrix $\hat{A}$, let $2\hat{m} = \sum_{i,j}\hat{A}_{ij}$. Let $\hat{k}_i =  \sum_{j}\hat{A}_{ij}$ be the sum of
the weights of the connections attached to the oriented link $i$. Suppose that for every node in $\hat{\mathcal{G}}$, $\hat{k}_i \leq \sqrt{2\hat{m}}$ $(\ast)$. In addition, assume in the first phrase, we iterate through links with the chosen orientations $\{e_0, e_1,...,e_N\}$ first, and then iterate through the links with the reverse orientations $\{\bar{e}_0,\bar{e}_
1,...,\bar{e}_N\}$. Then the Louvain method $\mathcal{A}$ will always cluster two orientations of a link $e$ and $\bar{e}$ into the same community. i.e. $\mathcal{A}(e) =\mathcal{A}(\bar{e}) $.
\label{thm:main}
\end{theorem}
The proof is in Appendix~\ref{app:pf}. Here $(\ast)$ is a sufficient but not necessary condition for the
theorem. We verified that the networks in our experiments
(Section~\ref{experiments}) indeed satisfy $(\ast)$  and thus Theorem 2 applies. One
can deliberately construct counterexamples that violate
$(\ast)$. However, this condition generally holds for large realistic
networks.

\begin{corollary}
Let $S(\hat{\mathcal{G}})$ be the line graph by considering the two orientations of the same link in $\hat{\mathcal{G}}$ as one supernode, and then aggregate the link weights accordingly. Then $S(\hat{\mathcal{G}})$ has adjacency matrix $\hat{A}_I+\hat{A}_{II}$. Under the assumption $(\ast)$, optimizing modularity on $S(\hat{\mathcal{G}})$ using the Louvain method is equivalent to optimizing modularity on $\hat{\mathcal{G}}$ using the Louvain method.
\label{cor:main}
\end{corollary}

\subsection{Spectral Properties of the Stochastic Lifting Matrix $\hat{\textbf{P}}$}\label{Phat}
Our assumption about the graph skeleton of $\mathcal{X}$ being undirected, connected and non-bipartite implies that $\G$ is also undirected, connected and non-bipartite. Hence the random walk on $\G$ defined by $\hat{\textbf{P}}$ is reversible~\cite{aldous-fill-2014}. The reversibility implies that $\hat{\textbf{P}}$, while not always being symmetric, is diagonalizable and has a real spectrum. We state the following interesting spectral properties about $\hat{\textbf{P}}$.
\begin{proposition}
$\hat{\textbf{P}}$ has the following block form:
\begin{equation*}
    \begin{bmatrix}
\hat{P}_I & \hat{P}_{II} \\
\hat{P}_{II} & \hat{P}_{I}
\end{bmatrix}
\end{equation*}
\label{prop:pblock}
\end{proposition}

\begin{proposition}
Let $\Spec(M)$ denote the spectrum of a matrix $M$. Then
$\Spec(\hat{\textbf{P}}) = \Spec(\hat{P}_I + \hat{P}_{II})\cup \Spec(\hat{P}_I - \hat{P}_{II})$.
\label{prop:decom}
\end{proposition}

\begin{proposition}
If $x$ is an eigenvector for $\hat{P}_I + \hat{P}_{II}$ with eigenvalue $\lambda$, then $[x,x]^\top$ is an eigenvector for $\hat{\textbf{P}}$ with eigenvalue $\lambda$.
\label{prop:even}
\end{proposition}

\begin{corollary}
$1$ is an eigenvalue for $\hat{P}_I + \hat{P}_{II}$, and suppose its corresponding eigenvector is $x$. Then $1$ is an eigenvalue for $\hat{\textbf{P}}$ with eigenvector $[x,x]^\top$.
\label{cor:even}
\end{corollary}

\begin{corollary}
Let the stationary solution of $\hat{\textbf{P}}$ be $\hat{\pi}$. Then $V^\top \hat{\pi} = 0$.
\label{cor:stationary}
\end{corollary}
The above properties of $\hat{\textbf{P}}$ imply that one can think of the random walk on the lifted line graph $\G$ as composed of two parts: the \emph{even} part $\Pone + \Ptwo$ and the \emph{odd} part $\Pone-\Ptwo$. We call the sum ``even'' because it suggests the symmetry between the two orientations in $\G$. As shown in Proposition~\ref{prop:even}, the eigenvalues of  $\hat{\textbf{P}}$ that come from the even part have eigenvectors that put the same value on both orientations for each link. Meanwhile, the difference is named ``odd'' because if one wishes to study a higher-order diffusion dynamic on SCs, one should consider the ``expectation process" defined in ~\cite{Parzanchevski2017SimplicialCS}, or equivalently, the ``lift-propagate-project'' procedure given by the definition of lifting~\eqref{def:lifting}: given link flows defined on $\X$, we first lift them to $\G$ via $V$, then the state is propagated through $\hat{\textbf{P}}$, and finally we project the lifted flows to the space on $\X$ via $\frac{1}{2}V^\top$~\cite{Schaub2020RandomWO}. The difference is exactly what remains after the projection, as the even part will cancel each other during the projection. 

\subsection{Method Summary and Computational Complexity}
Algorithm 1 summarizes the proposed higher-order link partitioning method. Given a simplicial complex $\X$, by definition, its boundary maps $B_1 \in \mathbb{R}^{n_0 \times n_1}$ and $B_2\in\mathbb{R}^{n_1 \times n_2}$ are sparse such that each column only contains $2$ and $3$ non-zero entries, respectively. Hence  computing $B_1$ and $B_2$ requires runtime $O(n_1)$ and $O(n_2)$, respectively and building $\A$ according to the formulas in Theorem~\ref{thm:lifting} runs in time $O(n_1+n_2)$. Moreover, the Louvain method is observed in practice to run in time $O(n_1 \log n_1)$ on average, although the exact computational complexity of the method is unknown~\cite{Blondel2008FastUO}. So the overall computational complexity of Algorithm 1 is $O(n_1 \log n_1 +n_2)$. Compared with running the Louvain method on a graph, the extra term $n_2$ is a trade-off between higher-order information and computational efficiency.
\begin{algorithm}
\caption{Higher-order link partition using $\mathcal{L}_1$ and lifting}
\textbf{Input}: a SC network $\X$, the Louvain method $\mathcal{A}$.\newline
\hspace*{2mm}1. Compute boundary maps $B_1$ and $B_2$ from $\X^1$ and $\X^2$.\newline
\hspace*{2mm}2. Build the adjacency matrix $\hat{A}$ for the lifted line graph $\G$.\newline
\hspace*{2mm}3. Feed $\hat{A}$ to $\mathcal{A}$ and obtain the final community structure $\hat{C}$ of $\G$.\newline
\hspace*{2mm}4. Under $(\ast)$, for $e \in \X^1$, $\mathcal{A}(e)$ = $\mathcal{A}(\bar{e})$. $\hat{C}$ can be directly projected to $\X$ and be interpreted as communities of links \hspace*{4.8mm} in $\X$.

\textbf{Output}: Link communities of $\X$ (communities of $\X^1$).
\end{algorithm}

\section{Experiments} \label{experiments}
In this section, we first use a simple synthetic example to give an intuitive illustration that our higher-order method can incorporate the higher-order connectivity information into partitioning---the reason why it works better than the lower-order methods. Then we evaluate the performance of our algorithm on real-world network data with higher-order information. Code is available at~\url{https://github.com/xinyiwu98/SC_link_comm}.

\begin{table}[h]
  \centering
  \caption{Datasets used in the experiments}
  \label{tab:datasets}
  \begin{tabular}{c|c|c|c}
    Name & Nodes & Links & Triangles\\
    \texttt{contact-high-school}\cite{Mastrandrea2015ContactPI} & 327 & 5,818 & 2,370\\
     \hline
    \texttt{contact-primary-school}\cite{Stehl2011HighResolutionMO} & 242 & 8,317 & 5,139 \\
     \hline
    \texttt{email-Enron}\cite{Benson2018SimplicialCA,Priebe2005ScanSO} & 144 & 1,344 & 1,159 \\
     \hline
    \texttt{email-Eu}\cite{Benson2018SimplicialCA,Yin2017LocalHG,Leskovec2007GraphED} & 979 & 29,299& 160,605 \\
     \hline
    \texttt{congress-bills}\cite{Fowler2006ConnectingTC,Benson2018SimplicialCA} &1712 & 66,102 & 86,164 \\
     \hline
    \texttt{senate-committees}\cite{Chodrow2021GenerativeHC,Stewart2005CharlesSI} &282 & 14,224 & 169,362 \\
     \hline
    \texttt{house-committees}\cite{Chodrow2021GenerativeHC,Stewart2005CharlesSI} &1290 & 126,155 & 2,996,327 \\
  \end{tabular}
\end{table}

\subsection{Baseline Methods}
In all the experiments, we compare our method (denoted by $\A$) with two well-known and commonly used types of graph-dimensional methods for link-based community detection: 
\begin{itemize}
    \item \textbf{Dendrogram cutting $S$}: Ahn et al. propose a dendrogram cutting method, which we denote as $S$. It is built on a similarity score and an objective called partition density~\cite{Ahn2010LinkCR}.
    \item \textbf{Modularity based approaches $C$, $D$, and $E_1$}: Evans and Lamboitte propose modularity optimization on three different weighted line graphs (corresponding adjacency matrices $C$, $D$ and $E_1$)~\cite{Evans2009LineGL}, and we use the Louvain method for the actual optimization in the experiments.
\end{itemize}

\vspace{-2ex}
\subsection{A Synthetic Illustrative Example}
 An implicit assumption of our proposed framework is that filled triangles in networks indicate tighter community structures. Consider the synthetic example shown in Figure \ref{fig:syn}A, where a group interaction occurs in the left triangle and the right triangle consists only of pairwise interactions. We show that our higher-order method $\A$ better clusters links by distinguishing the filled and the open triangles.
 
 We first ran the Louvain method with increasing time step $t$ on the lifted line graph represented by $\A$ with the triangle information. We then ran the Louvain method on the weighted line graphs represented by $C$, $D$ and $E_1$, for which we ignored the triangle information. Similarly, for the other baseline $S$, we ignored the higher-order connection and built the dendrogram based on the graph. We specified our desired number of link communities to be $2$. We ran all the methods until they output two communities and repeated for 30 times. If the optimal partition strictly favored a single community over two communities, we used the single community as the result. Figure \ref{fig:syn}B,C present a summary of link clustering results for all methods, where we observed four community patterns. For each pattern, each color represents a distinct community.  
 
 Our method $\A$ is built on both lower-order and higher-order connectivity and yielded the most reasonable pattern $4$ consistently. $C$, $E_1$, and $S$ all failed to identify pattern $4$. $D$ output pattern $4$ around $50\%$ of runs, due to the randomness of the Louvain method. However, the lack of higher-order information of $D$ led to inconsistent results between pattern $2$ and $4$, as its model cannot distinguish the two patterns. The results suggest that higher-order information indeed facilitates finding closely interrelated links.

 \begin{figure}[h]
 \centering
     \includegraphics[width = 8cm]{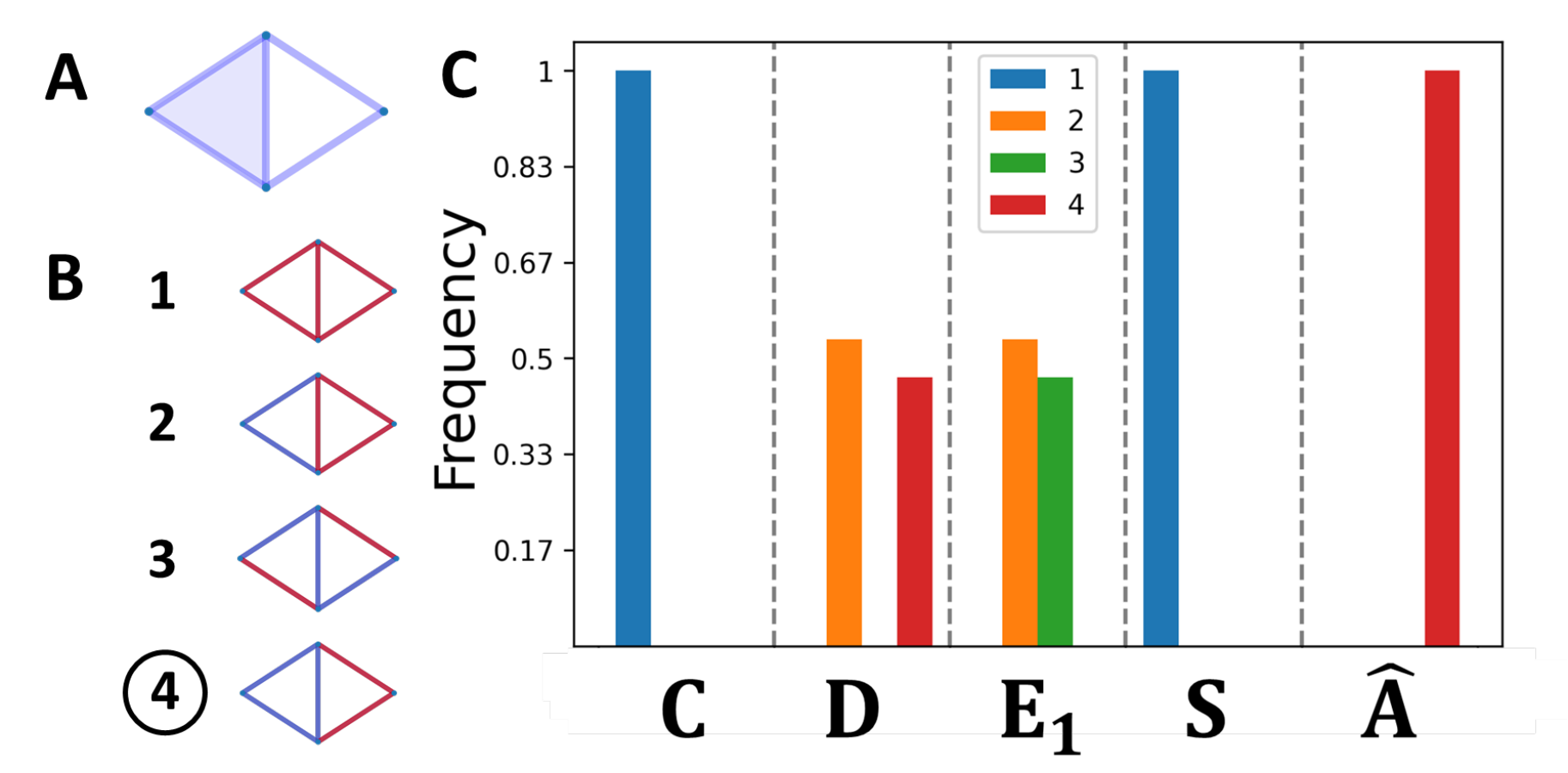}
    \caption{Link partitioning results of \textbf{A}, where the shaded area represents a 2-simplex; \textbf{B}: we observe $4$ partition patterns in total; \textbf{C}: summary of partitioning results for each method over $30$ runs.  Only our method $\A$, which captures the
    existence of the higher-order connection, produced the most reasonable pattern 4 consistently over all runs. }
    \label{fig:syn}
\end{figure}

\subsection{Datasets and Evaluation Measures}
We now compare our higher-order framework with graph-based baselines on real-world networks with higher-order information. We focus on social networks where filled triangles have expressive power and clear interpretations~\cite{Holland1970AMF,Granovetter1973TheSO,Newman2001ClusteringAP}. Summary statistics of the datasets used are present in Table \ref{tab:datasets}. See Appendix~\ref{app:data} for detailed descriptions for each dataset and its metadata used accordingly.

 As Ahn et al. point out, there are multiple aspects that one must consider when comparing different community detection algorithms such as the quality of the found communities, how much of the network is studied (coverage), and whether the algorithm finds overlapping structures~\cite{Ahn2010LinkCR}. Therefore, it is hard to justify comparing methods only along a single dimension and claim one method outperforms the rest. In fact, no algorithm can be universally optimal for all possible community detection tasks simultaneously~\cite{Peel2017TheGT}. Thus we adopted the composite performance measure proposed in \cite{Ahn2010LinkCR} to account for the different focuses of each method. The composite measure is comprised of four metrics measuring four distinct aspects of the detected communities. For each metric, the higher the value, the better the partitioning result on what it measures.
 The four metrics are as follows.
\paragraph{\textbf{Community quality}} Each dataset possesses metadata that describe each node. We assume that nodes in the same communities share more metadata than nodes in different ones. We then compute the \emph{enrichment score of node pair}:\vspace{-1ex}$$Enrichment = \frac{ \overline{s(i,j)}_{\text{i,j within same communities}}}{\overline{s(i,j)}_{\text{all pairs i,j}}}\,,$$
where $s(i,j)$ is a metadata-based similarity score between node $i$ and $j$, and $\overline{\cdot}$ denotes the average over the particular group. The larger the enrichment, the ``tighter'' the communities according to the metadata.
\paragraph{\textbf{Overlap quality}} 
We use overlap quality to measure how much information about the true overlap of nodes is learned from the link communities. Each dataset possesses additional metadata that are closely related to the number of true communities that each node participates in. To rigorously quantify the amount of information gained about the overlapping structures, we calculate the mutual information between the number of memberships of each node given by the link communities and the metadata. 
\paragraph{\textbf{Community coverage}} Community coverage measures how much of the network is studied. We count the fraction of nodes that belong to at least one nontrivial community given by the link partition, \textit{i.e.} community with three or more nodes.
\paragraph{\textbf{Overlap coverage}} Two algorithms may both give complete coverage, but one may extract more information by finding more overlapping communities. We measure how densely the found communities overlap by counting the average number of memberships in nontrivial communities that nodes are given by the link partition.

\subsection{Comparison Results on Modularity }
For the line graph methods including ours, we looked at the resulting communities at their own optimal modularity levels. For the other method $S$, we looked at the
communities at its optimal objective level. For each dataset, we used $100$ networks such that each network had higher-order information encoded by $2$-simplices. Since there is no \textit{a priori} knowledge how closely these different objectives align and how well the Louvain method is heuristically optimizing modularity in our case, we reported the modularity $Q(\A)$ of link communities found by different methods.  Comparison results in terms of  modularity $Q(\A)$ of link communities are presented in Table~\ref{tab:modularity}. For all the networks, the algorithms indeed found different solutions. The solutions also accorded with our expectation that modularity $Q(\A)$ was higher for our method. Although optimizing modularity is NP-hard~\cite{Brandes2006MaximizingMI}, the Louvain method does find good local optimal solutions with high modularity $Q(\A)$.

\begin{table*}[t]
\scriptsize
  \caption{Comparison results of \textbf{modularity $Q(\A)$} of the link communities. The modularity scores are averages over the $100$ networks. Bold denotes the highest result. }
  \label{tab:modularity}
   \begin{tabular}{c|c|c|c|c|c|c}
    Dataset & $\#$ networks (sizes) & line graph $C$ & line graph $D$ & line graph $E_1$ & dendrogram $S$ & higher-order $\hat{A}$\\
    \texttt{contact-high-school} & 100 (289-554)  & 0.723$\pm$\tiny{0.002} & 0.700$\pm$\tiny{0.003} & 0.669$\pm$\tiny{0.004}&
    0.669$\pm$\tiny{0.003}&\textbf{0.752}$\pm$\tiny{0.001}\\
     \hline
    \texttt{contact-primary-school} & 100 (720-1210)  & 0.670$\pm$\tiny{0.001} & 0.612$\pm$\tiny{0.001} & 0.531$\pm$\tiny{0.003}&
    0.406$\pm$\tiny{0.017}&
    \textbf{0.702}$\pm$\tiny{0.000}\\
     \hline
     \texttt{email-Enron}& 100  (218-430) & 0.702$\pm$\tiny{0.001} & 0.635$\pm$\tiny{0.002} & 0.594$\pm$\tiny{0.004}&
    0.578$\pm$\tiny{0.006}&
    \textbf{0.723}$\pm$\tiny{0.001}\\
     \hline
    \texttt{email-Eu} & 100 (165-409) & 0.708$\pm$\tiny{0.002} & 0.643$\pm$\tiny{0.003} & 0.609$\pm$\tiny{0.005}&
    0.565$\pm$\tiny{0.008}&\textbf{0.732}$\pm$\tiny{0.002}\\
     \hline
    \texttt{congress-bills} & 100 (142-333)& 0.718$\pm$\tiny{0.003} & 0.657$\pm$\tiny{0.003} & 0.622$\pm$\tiny{0.005} &
    0.598$\pm$\tiny{0.005} &\textbf{0.750}$\pm$\tiny{0.001}\\
     \hline
    \texttt{senate-committees} & 100 (1005-1670)& 0.670$\pm$\tiny{0.001} & 0.612$\pm$\tiny{0.001} & 0.531$\pm$\tiny{0.003} &
    0.208$\pm$\tiny{0.017} &\textbf{0.702}$\pm$\tiny{0.000}\\
     \hline
    \texttt{house-committees} & 100 (540-886)& 0.707$\pm$\tiny{0.001} & 0.658$\pm$\tiny{0.001} & 0.595$\pm$\tiny{0.003} &
    0.565$\pm$\tiny{0.010} &\textbf{0.731}$\pm$\tiny{0.001}\\
  \end{tabular}
  \centering
\end{table*}

\subsection{Comparison Results on Composite Performance}

For each dataset, we examined the same $100$ networks used in Table~\ref{tab:modularity} whose sizes ranged between $59-93$ nodes. The results for community quality and overlap quality are presented in Table ~\ref{tab:community_quality} and Table~\ref{tab:overlap}, respectively. 

\begin{table*}[h]\scriptsize\centering
  \caption{Comparison results of \textbf{overlap quality} given by the link communities. The scores are averages over the 100 networks.  The best result is highlighted in bold.}
  \label{tab:overlap}
  \begin{tabular}{c|c|c|c|c|c|c}
    Dataset & $\#$ networks (sizes) & line graph $C$ & line graph $D$ & line graph $E_1$ & dendrogram $S$ & higher-order $\hat{A}$\\
    \texttt{contact-high-school} & 100 (76-93) & 1.230$\pm$\tiny{0.014} & 1.164$\pm$\tiny{0.011} & 0.951$\pm$\tiny{0.010}& 1.519$\pm$\tiny{0.016}& \textbf{1.863}$\pm$\tiny{0.012}\\
     \hline
     \texttt{contact-primary-school} & 100 (72-91) & 2.259$\pm$\tiny{0.012} & 2.037$\pm$\tiny{0.012} & 1.377$\pm$\tiny{0.015}& 1.904$\pm$\tiny{0.050}&\textbf{2.816}$\pm$\tiny{0.0010}\\
      \hline
    \texttt{email-Enron}  & 100 (59-78) & 1.585$\pm$\tiny{0.015} & 1.335$\pm$\tiny{0.014} & 1.069$\pm$\tiny{0.012}& 1.782$\pm$\tiny{0.016}&\textbf{1.993}$\pm$\tiny{0.014}\\
     \hline
    \texttt{email-Eu}  & 100 (61-86) & 1.548$\pm$\tiny{0.022} & 1.146$\pm$\tiny{0.015} & 1.006$\pm$\tiny{0.013}& 1.712$\pm$\tiny{0.018}&\textbf{1.840}$\pm$\tiny{0.019}\\
     \hline
    \texttt{congress-bills}  & 100 (68-91) & 0.283$\pm$\tiny{0.007} & 0.289$\pm$\tiny{0.006} & 0.216$\pm$\tiny{0.006}& 0.350$\pm$\tiny{0.009}&\textbf{0.442}$\pm$\tiny{0.008}\\
     \hline
    \texttt{senate-committees}  & 100 (73-92) & 1.096 $\pm$\tiny{0.011} & 0.927$\pm$\tiny{0.014} & 0.420$\pm$\tiny{0.009}& 0.725$\pm$\tiny{0.029}&\textbf{1.505}$\pm$\tiny{0.012}\\
     \hline
    \texttt{house-committees}  & 100 (91-99) & 0.627 $\pm$\tiny{0.001} & 0.691$\pm$\tiny{0.009} & 0.401$\pm$\tiny{0.008}& 0.797$\pm$\tiny{0.013}&\textbf{0.905}$\pm$\tiny{0.011}\\
  \end{tabular}
    \centering
\end{table*}

In all of the $7$ datasets, our method achieved the best overlap quality. On average, our method achieved an improvement of $20.5\%$ with respect to the second best method. In addition, the gain in overlap quality compared to $C$, $D$ and $E_1$ was across different time step $t$. The first row of Figure~\ref{fig:overlap} shows the comparison results of overlap quality on three networks built from the entire corresponding datasets, \texttt{email-Enron}, \texttt{contact-high-school} and \texttt{contact-primary-school}. For all the methods, overlap quality overall decreases as time step $t$ increases because the communities get coarser. Nonetheless, compared to $C$, $D$ and $E_1$, the gain in overlap quality using higher-order information is consistent across random walk time steps $t$.

\begin{table*}[t]\scriptsize\centering
  \caption{Comparison results of \textbf{community quality} given by link communities. The scores are averages over the 100 networks. The best result is highlighted in bold.}
  \label{tab:community_quality}
  \begin{tabular}{c|c|c|c|c|c|c}
    Dataset & $\#$ networks (sizes) & line graph $C$ & line graph $D$ & line graph $E_1$ & dendrogram $S$ & higher-order $\hat{A}$\\
    \texttt{contact-high-school} & 100 (76-93) & 2.082$\pm$\tiny{0.028} & 2.670$\pm$\tiny{0.030} & 2.350$\pm$\tiny{0.031}& \textbf{5.880}$\pm$\tiny{0.061}& 2.642$\pm$\tiny{0.033}\\
     \hline
     \texttt{contact-primary-school} & 100 (72-91) & 
     1.073$\pm$\tiny{0.003} & 1.304$\pm$\tiny{0.005} & 1.240$\pm$\tiny{0.006}& \textbf{1.942}$\pm$\tiny{0.053}&1.142$\pm$\tiny{0.004}\\
      \hline
    \texttt{email-Enron}  & 100 (59-78) & 1.141$\pm$\tiny{0.007} & 1.200$\pm$\tiny{0.010} & 1.224$\pm$\tiny{0.010}& \textbf{1.514}$\pm$\tiny{0.021}&1.217$\pm$\tiny{0.009}\\
     \hline
    \texttt{email-Eu}  & 100 (61-86) & 1.774$\pm$\tiny{0.029} & 2.547$\pm$\tiny{0.041} & 2.301$\pm$\tiny{0.031}& \textbf{5.030}$\pm$\tiny{0.127}&
    2.416$\pm$\tiny{0.052}\\
     \hline
    \texttt{congress-bills}  & 100 (68-91) & 1.008$\pm$\tiny{0.003} & 1.010$\pm$\tiny{0.003} & 1.012$\pm$\tiny{0.003}& \textbf{1.028}$\pm$\tiny{0.006}&1.012$\pm$\tiny{0.004}\\
     \hline
    \texttt{senate-committees}  & 100 (73-92) & 0.999$\pm$\tiny{0.000} & 0.996$\pm$\tiny{0.001} & 0.998$\pm$\tiny{0.001}& \textbf{1.000}$\pm$\tiny{0.001}&\textbf{1.000}$\pm$\tiny{0.000}\\
     \hline
    \texttt{house-committees}  & 100 (91-99) & 1.001$\pm$\tiny{0.001} & \textbf{1.003}$\pm$\tiny{0.002} & 1.000$\pm$\tiny{0.001}& 1.001$\pm$\tiny{0.003}&1.001$\pm$\tiny{0.002}\\
  \end{tabular}
  \centering
\end{table*}

\begin{figure}[b]
    \centering
    \includegraphics[width = 10cm]{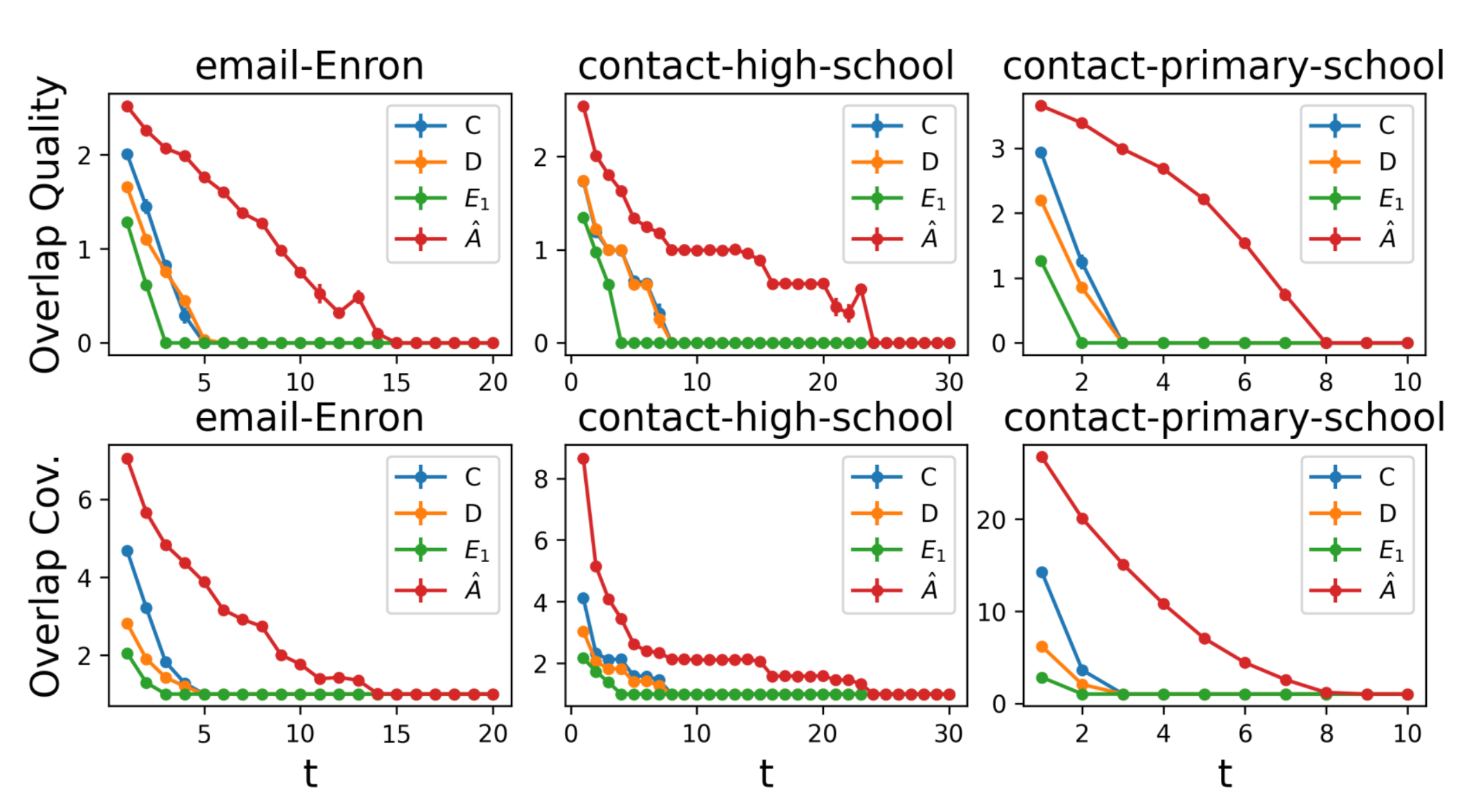}
    \caption{Comparison of overlap quality (top row) and overlap coverage (bottom row) of resulting partitions across random walk time steps $t$ for \texttt{email-Enron}, \texttt{contact-high-school} and \texttt{contact-primary-school}. Our method (red) consistently provides the highest values and outperforms the line graph baseline methods ($C$, $D$, $E_1$) at all time steps $t$ for both metrics.}
\label{fig:overlap}
\vspace{-2ex}
\end{figure}

In $6$ of the $7$ datasets, the baseline method $S$ achieved the best community quality. All the line graph methods including ours had rather similar performances. We believe that this is because of the disconnection between the community structure given by link communities of the line graph with high modularity and the chosen community metadata. Often, community structures do not align with the observed metadata~\cite{Peel2017TheGT}. However, the fact that an algorithm fails to recover a community structure close to the metadata does not imply that
it is “failing” in its stated objective, namely, modularity optimization of the line graph. The algorithm is discovering different but still valuable community structures.

We then computed community coverage and overlap coverage accordingly for all the networks in every dataset. The results are shown in Table~\ref{tab:coverage}. Notably, our higher-order method outperformed all the baseline methods in terms of both coverage metrics. For overlap coverage, our method achieved an improvement of $47.0\%$ with respect to the second best method. Like the gain in overlap quality, the improvement in overlap coverage compared to the line graph baselines was also consistent across all time steps $t$. See the second row of Figure \ref{fig:overlap} for the comparison results of overlap coverage on the entire networks of \texttt{email-Enron}, \texttt{contact-high-school}, and \texttt{contact-primary-school}.

\begin{figure*}[t]
\centering
    \includegraphics[width=\textwidth]{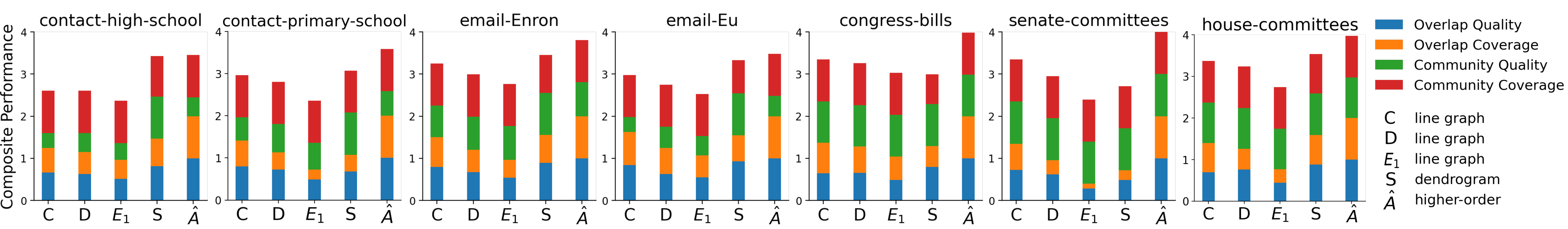}
\caption{Comparison results of composite performance of link-based community detection methods on real world networks. Higher-order link clustering ($\A$) finds community structures with overall better quality and higher-coverage in all datasets. }
\label{fig:composite}
\end{figure*}

Finally, in order to compare the composite performance on each dataset, we followed the procedure in Ahn et al.~\cite{Ahn2010LinkCR} and renormalized all values for each metric such that the maximum value is $1$ for the best performing method. Figure~\ref{fig:composite} displays the comparison results of composite performance on all datasets. Our higher-order method outperformed the graph-based methods on the composite performance in all the datasets, with an average improvement of $12.0\%$ over the second best score. The results show that like many other problems in network science, higher-order link-based community detection reveals more relevant and more complex community structures in networks as it captures the valuable higher-order information.  

\begin{table*}[h]\scriptsize\centering
  \caption{Comparison results of \textbf{community coverage} and \textbf{overlap coverage} given by the link communities. The scores are averages over the 100 networks. The best result is highlighted in bold. }
  \label{tab:coverage}

  \begin{tabular}{c|c|c|c|c|c|c|c}
     Dataset & $\#$ networks (sizes) & cov. & line graph $C$ & line graph $D$ & line graph $E_1$ & dendrogram $S$ & higher-order $\hat{A}$\\
     \multirow{ 2}{*}{\texttt{contact-high-school}} &\multirow{ 2}{*}{100 (76-93)}& comm &\textbf{1.000}$\pm$\tiny{0.000}       &\textbf{1.000}$\pm$\tiny{0.000}    &\textbf{1.000}$\pm$\tiny{0.000}    &0.959$\pm$\tiny{0.003}
     &\textbf{1.000}$\pm$\tiny{0.000}       \\  
    & & overlap  &
    2.003$\pm$\tiny{0.025} & 1.784$\pm$\tiny{0.014}&  1.544$\pm$\tiny{0.010}  &2.222$\pm$\tiny{0.029}
    &\textbf{3.417}$\pm$\tiny{0.037}       \\  
    \hline
    \multirow{ 2}{*}{\texttt{contact-primary-school}} &\multirow{ 2}{*}{100 (72-91)}& comm &\textbf{1.000}$\pm$\tiny{0.000}       &\textbf{1.000}$\pm$\tiny{0.000}    &\textbf{1.000}$\pm$\tiny{0.000}    &0.995$\pm$\tiny{0.001}
     &\textbf{1.000}$\pm$\tiny{0.000}       \\  
    & & overlap  &
    6.255$\pm$\tiny{0.074} & 4.196$\pm$\tiny{0.051}&  2.382$\pm$\tiny{0.024}  &4.080$\pm$\tiny{0.138}
    &\textbf{10.221}$\pm$\tiny{0.096}       \\  
    \hline
    \multirow{ 2}{*}{\texttt{email-Enron}} &\multirow{ 2}{*}{100 (59-78)}& comm &\textbf{1.000}$\pm$\tiny{0.000}       &\textbf{1.000}$\pm$\tiny{0.000}    &\textbf{1.000}$\pm$\tiny{0.000}    &0.903$\pm$\tiny{0.005}
     &\textbf{1.000}$\pm$\tiny{0.000}       \\  
    & & overlap  &
    2.820$\pm$\tiny{0.037} & 2.122$\pm$\tiny{0.022}&  1.695$\pm$\tiny{0.014}  &2.640$\pm$\tiny{0.038}
    &\textbf{4.017}$\pm$\tiny{0.047}       \\  
    \hline
    \multirow{ 2}{*}{\texttt{email-Eu}} &\multirow{ 2}{*}{100 (61-86)}& comm &\textbf{1.000}$\pm$\tiny{0.000}       &\textbf{1.000}$\pm$\tiny{0.000}    &\textbf{1.000}$\pm$\tiny{0.000}    &0.783$\pm$\tiny{0.006}
     &\textbf{1.000}$\pm$\tiny{0.000}       \\  
    & & overlap  &
    2.456$\pm$\tiny{0.043} & 1.939$\pm$\tiny{0.025}&  1.639$\pm$\tiny{0.014}  &1.918$\pm$\tiny{0.036}
    &\textbf{3.137}$\pm$\tiny{0.053}       \\  
    \hline
     \multirow{ 2}{*}{\texttt{congress-bills}} &\multirow{ 2}{*}{100 (68-91)}& comm &\textbf{1.000}$\pm$\tiny{0.000}       &\textbf{1.000}$\pm$\tiny{0.000}    &\textbf{1.000}$\pm$\tiny{0.000}    &0.705$\pm$\tiny{0.007}
     &\textbf{1.000}$\pm$\tiny{0.000}       \\  
    & & overlap  &
    1.794$\pm$\tiny{0.022} & 1.546$\pm$\tiny{0.017}&  1.374$\pm$\tiny{0.011}  &1.230$\pm$\tiny{0.022}
    &\textbf{2.479}$\pm$\tiny{0.032}       \\  
    \hline
     \multirow{ 2}{*}{\texttt{senate-committees}} &\multirow{ 2}{*}{100 (73-92)}& comm &\textbf{1.000}$\pm$\tiny{0.000}       &\textbf{1.000}$\pm$\tiny{0.000}    &\textbf{1.000}$\pm$\tiny{0.000}    &0.999$\pm$\tiny{0.000}
     &\textbf{1.000}$\pm$\tiny{0.000}       \\  
    & & overlap  &
    8.478$\pm$\tiny{0.118} & 4.679$\pm$\tiny{0.099}&  1.645$\pm$\tiny{0.017}  &3.147$\pm$\tiny{0.138}
    &\textbf{13.738}$\pm$\tiny{0.153}       \\  
    \hline
     \multirow{ 2}{*}{\texttt{house-committees}} &\multirow{ 2}{*}{100 (91-99)}& comm &\textbf{1.000}$\pm$\tiny{0.000}       &\textbf{1.000}$\pm$\tiny{0.000}    &\textbf{1.000}$\pm$\tiny{0.000}    &0.946$\pm$\tiny{0.003}
     &\textbf{1.000}$\pm$\tiny{0.000}       \\  
    & & overlap  &
     4.001$\pm$\tiny{0.059} & 2.814$\pm$\tiny{0.035}&  1.837$\pm$\tiny{0.021}  &3.999$\pm$\tiny{0.075}
    &\textbf{5.661}$\pm$\tiny{0.083}       \\  
  \end{tabular}

\end{table*}

\section{Discussion}
In this paper, we use a novel random walk on links of simplicial complexes defined according to higher-order Laplacians to perform link partitioning that incorporates higher-order network information. We find that using higher-order information provides substantial improvements in discovering overlapping community structure through links. These results emphasize the value of higher-order network data for studying the structure and behavior of networks and necessitates higher-order network models that can take account of this information. Since our method connects to the spectral theory for simplicial complexes through random walks, future research should aim to study this relationship more.

\section{Acknowledgement}
Xinyi Wu would like to thank Zihui Wu for inspiring discussions. This research has been supported in part by ARO grant number W911NF-19-1-0057, and a Vannevar Bush Fellowship from the Office of the Secretary of Defense.

\bibliographystyle{plain}
\bibliography{IEEEabrv}  

\begin{thebibliography}{10}

\bibitem{Ahn2010LinkCR}
Yong-Yeol Ahn, James~P. Bagrow, and Sune Lehmann.
\newblock Link communities reveal multiscale complexity in networks.
\newblock {\em Nature}, 2010.

\bibitem{aldous-fill-2014}
David Aldous and James~Allen Fill.
\newblock Reversible markov chains and random walks on graphs, 2002.

\bibitem{Barbarossa2020TopologicalSP}
Sergio Barbarossa and Stefania Sardellitti.
\newblock Topological signal processing over simplicial complexes.
\newblock {\em IEEE Trans. Signal Process}, 2020.

\bibitem{Benson2018SimplicialCA}
Austin~R. Benson, Rediet Abebe, Michael Schaub, Ali Jadbabaie, and Jon
  Kleinberg.
\newblock Simplicial closure and higher-order link prediction.
\newblock {\em PNAS}, 2018.

\bibitem{Benson2021HigherorderNA}
Austin~R. Benson, David~F. Gleich, and Desmond~J. Higham.
\newblock Higher-order network analysis takes off, fueled by classical ideas
  and new data.
\newblock {\em ArXiv}, 2103.05031, 2021.

\bibitem{Benson2015TensorSC}
Austin~R. Benson, David~F. Gleich, and Jure Leskovec.
\newblock Tensor spectral clustering for partitioning higher-order network
  structures.
\newblock In {\em SDM}, 2015.

\bibitem{Benson2016HigherorderOO}
Austin~R. Benson, David~F. Gleich, and Jure Leskovec.
\newblock Higher-order organization of complex networks.
\newblock {\em Science}, 2016.

\bibitem{Blondel2008FastUO}
Vincent~D. Blondel, Jean-Loup Guillaume, Renaud Lambiotte, and Etienne
  Lefebvre.
\newblock Fast unfolding of communities in large networks.
\newblock {\em J. Stat. Mech.}, 2008.

\bibitem{Brandes2006MaximizingMI}
Ulrik Brandes, Daniel Delling, Marco Gaertler, Rachelle Goerke, Martin Hoefer,
  Zoran Nikoloski, and Donald Wagner.
\newblock Maximizing modularity is hard.
\newblock {\em ArXiv}, physics/0608255, 2006.

\bibitem{Chodrow2021GenerativeHC}
Philip~S. Chodrow, Nate Veldt, and Austin~R. Benson.
\newblock Generative hypergraph clustering: From blockmodels to modularity.
\newblock {\em Science Advances}, 2021.

\bibitem{Chung1996SpectralGT}
Fan R.~K. Chung.
\newblock Spectral graph theory.
\newblock 1996.

\bibitem{Coscia2012DEMONAL}
Michele Coscia, Giulio Rossetti, Fosca Giannotti, and Dino Pedreschi.
\newblock Demon: a local-first discovery method for overlapping communities.
\newblock In {\em KDD}, 2012.

\bibitem{Delvenne12755}
J.-C. Delvenne, S.~N. Yaliraki, and M.~Barahona.
\newblock Stability of graph communities across time scales.
\newblock {\em PNAS}, 2010.

\bibitem{Deng2017FindingOC}
X.~Deng, G.~Li, M.~Dong, and K.~Ota.
\newblock Finding overlapping communities based on markov chain and link
  clustering.
\newblock {\em Peer-to-Peer Netw. Appl.}, 2017.

\bibitem{Ebli2019ANO}
Stefania Ebli and Gard Spreemann.
\newblock A notion of harmonic clustering in simplicial complexes.
\newblock In {\em ICMLA}, 2019.

\bibitem{Eckmann1944HarmonischeFU}
Beno Eckmann.
\newblock Harmonische funktionen und randwertaufgaben in einem komplex.
\newblock {\em Commentarii Mathematici Helvetici}, 17:240--255, 1944.

\bibitem{Evans2009LineGL}
Tim~S. Evans and Renaud Lambiotte.
\newblock Line graphs, link partitions, and overlapping communities.
\newblock {\em Phys. Rev. E Stat. Nonlin. Soft Matter Phys.}, 2009.

\bibitem{Evans2010LineGO}
Tim~S. Evans and Renaud Lambiotte.
\newblock Line graphs of weighted networks for overlapping communities.
\newblock {\em The European Physical Journal B}, 2010.

\bibitem{FORTUNATO201075}
Santo Fortunato.
\newblock Community detection in graphs.
\newblock {\em Physics Reports}, 2010.

\bibitem{Fowler2006ConnectingTC}
James~H. Fowler.
\newblock Connecting the congress: A study of cosponsorship networks.
\newblock {\em Political Analysis}, 2006.

\bibitem{Glaze2021PrincipledSN}
Nicholas Glaze, T.~Mitchell Roddenberry, and Santiago Segarra.
\newblock Principled simplicial neural networks for trajectory prediction.
\newblock In {\em ICML}, 2021.

\bibitem{Gopalan2013EfficientDO}
Prem Gopalan and David~M. Blei.
\newblock Efficient discovery of overlapping communities in massive networks.
\newblock {\em PNAS}, 2013.

\bibitem{Granovetter1973TheSO}
Mark~S. Granovetter.
\newblock The strength of weak ties.
\newblock {\em Am. J. Sociol.}, 1973.

\bibitem{Hatcher2002AlgebraicT}
Allen Hatcher.
\newblock Algebraic topology.
\newblock 2002.

\bibitem{Holland1970AMF}
Paul Holland and Samuel Leinhardt.
\newblock A method for detecting structure in sociometric data.
\newblock {\em Am. J. Sociol.}, 1970.

\bibitem{Horak2011SpectraOC}
Danijela Horak and J{\"u}rgen Jost.
\newblock Spectra of combinatorial laplace operators on simplicial complexes.
\newblock {\em Advances in Mathematics}, 2011.

\bibitem{Kaufman2018HighOR}
Tali Kaufman and Izhar Oppenheim.
\newblock High order random walks: Beyond spectral gap.
\newblock {\em Combinatorica}, 2018.

\bibitem{Lambiotte2008LaplacianDA}
R.~Lambiotte, J.-C. Delvenne, and M.~Barahona.
\newblock Laplacian dynamics and multiscale modular structure in networks.
\newblock {\em IEEE Trans. Netw. Sci. Eng.}, 2008.

\bibitem{Lancichinetti2009DetectingTO}
A.~Lancichinetti, S.~Fortunato, and J.~Kert{\'e}sz.
\newblock Detecting the overlapping and hierarchical community structure in
  complex networks.
\newblock {\em New J. Phys}, 2009.

\bibitem{Lee2017InverseRL}
Juyong Lee, Zhongyuan Zhang, Jooyoung Lee, Bernard~R. Brooks, and Yong-Yeol
  Ahn.
\newblock Inverse resolution limit of partition density and detecting
  overlapping communities by link-surprise.
\newblock {\em Scientific Reports}, 2017.

\bibitem{Leskovec2007GraphED}
Jure Leskovec, Jon~M. Kleinberg, and Christos Faloutsos.
\newblock Graph evolution: Densification and shrinking diameters.
\newblock {\em ACM Trans. Knowl. Discov. Data}, 2007.

\bibitem{Lim2020HodgeLO}
Lek-Heng Lim.
\newblock Hodge laplacians on graphs.
\newblock {\em ArXiv}, 1507.05379, 2020.

\bibitem{Mastrandrea2015ContactPI}
Rossana Mastrandrea, Julie Fournet, and Alain Barrat.
\newblock Contact patterns in a high school: A comparison between data
  collected using wearable sensors, contact diaries and friendship surveys.
\newblock {\em PLoS ONE}, 2015.

\bibitem{Medaglia2015CognitiveNN}
John~D. Medaglia, Mary-Ellen Lynall, and Danielle~S. Bassett.
\newblock Cognitive network neuroscience.
\newblock {\em J. Cogn. Neurosci.}, 2015.

\bibitem{Milo2002NetworkMS}
Ron Milo, Shai~S. Shen-Orr, Shalev Itzkovitz, Nadav Kashtan, Dmitri~B.
  Chklovskii, and Uri Alon.
\newblock Network motifs: simple building blocks of complex networks.
\newblock {\em Science}, 2002.

\bibitem{Mukherjee2016RandomWO}
Sayan Mukherjee and John Steenbergen.
\newblock Random walks on simplicial complexes and harmonics.
\newblock {\em Random Struct. Algorithms}, 2016.

\bibitem{Newman2001ClusteringAP}
Mark E.~J. Newman.
\newblock Clustering and preferential attachment in growing networks.
\newblock {\em Phys. Rev. E Stat. Nonlin. Soft Matter Phys.}, 2001.

\bibitem{Newman2006ModularityAC}
Mark E.~J. Newman.
\newblock Modularity and community structure in networks.
\newblock {\em PNAS}, 2006.

\bibitem{Orgaz2018AMG}
Gema~Bello Orgaz, Sancho Salcedo-Sanz, and David Camacho.
\newblock A multi-objective genetic algorithm for overlapping community
  detection based on edge encoding.
\newblock {\em Inf. Sci.}, 2018.

\bibitem{Palla2005UncoveringTO}
G.~Palla, I.~Der{\'e}nyi, I.~J. Farkas, and T.~Vicsek.
\newblock Uncovering the overlapping community structure of complex networks in
  nature and society.
\newblock {\em Nature}, 2005.

\bibitem{Parzanchevski2017SimplicialCS}
Ori Parzanchevski and Ron Rosenthal.
\newblock Simplicial complexes: Spectrum, homology and random walks.
\newblock {\em Random Struct. Algorithms}, 2017.

\bibitem{Parzanchevski2016IsoperimetricII}
Ori Parzanchevski, Ron Rosenthal, and Ran~J. Tessler.
\newblock Isoperimetric inequalities in simplicial complexes.
\newblock {\em Combinatorica}, 2016.

\bibitem{Peel2017TheGT}
Leto Peel, Daniel~B. Larremore, and Aaron Clauset.
\newblock The ground truth about metadata and community detection in networks.
\newblock {\em Science Advances}, 2017.

\bibitem{Priebe2005ScanSO}
Carey~E. Priebe, John~M. Conroy, David~J. Marchette, and Youngser Park.
\newblock Scan statistics on enron graphs.
\newblock {\em Comput. Math. Organ. Theory}, 2005.

\bibitem{Sarker2021HigherOI}
Arnab Sarker, Jean-Baptiste Seby, Austin~R. Benson, and Ali Jadbabaie.
\newblock Higher order information identifies tie strength.
\newblock {\em ArXiv}, abs/2108.02091, 2021.

\bibitem{Schaub2017TheMF}
M.~T. Schaub, J.-C. Delvenne, M.~Rosvall, and R.~Lambiotte.
\newblock The many facets of community detection in complex networks.
\newblock {\em Appl. Netw. Sci.}, 2017.

\bibitem{Schaub2020RandomWO}
Michael~T. Schaub, Austin~R. Benson, Paul Horn, G{\'a}bor Lippner, and Ali
  Jadbabaie.
\newblock Random walks on simplicial complexes and the normalized hodge
  1-laplacian.
\newblock {\em SIAM Review}, 2020.

\bibitem{Shi2013ALC}
Chuan Shi, Yanan Cai, Di~Fu, Yuxiao Dong, and Bin Wu.
\newblock A link clustering based overlapping community detection algorithm.
\newblock {\em Data Knowl. Eng.}, 2013.

\bibitem{Sotiropoulos2021TriangleawareSS}
Konstantinos Sotiropoulos and Charalampos~E. Tsourakakis.
\newblock Triangle-aware spectral sparsifiers and community detection.
\newblock In {\em KDD}, 2021.

\bibitem{Stehl2011HighResolutionMO}
J.~Stehl{\'e}, N.~Voirin, A.~Barrat, C.~Cattuto, L.~Isella, J.-F. Pinton,
  M.~Quaggiotto, W.~Van den Broeck, C.~R{\'e}gis, B.~Lina, and P.~Vanhems.
\newblock High-resolution measurements of face-to-face contact patterns in a
  primary school.
\newblock {\em PLoS ONE}, 2011.

\bibitem{Stewart2005CharlesSI}
Charles Stewart and Jonathan Woon.
\newblock Congressional committee assignments, 103rd to 105th congresses,
  1993--1998, data.
\newblock 2005.

\bibitem{Yin2017LocalHG}
Hao Yin, Austin~R. Benson, Jure Leskovec, and David~F. Gleich.
\newblock Local higher-order graph clustering.
\newblock In {\em KDD}, 2017.

\bibitem{Zhang2015SymmetricNM}
Xiang Zhang, Naiyang Guan, Wenju Zhang, Xuhui Huang, Shuyi Wu, and Zhigang Luo.
\newblock Symmetric non-negative matrix factorization based link partition
  method for overlapping community detection.
\newblock In {\em IEEE SMC}, 2015.

\end{thebibliography}






\appendix

\section{Proofs}\label{app:pf}
\begin{proof}[Proof of Theorem~\ref{thm:lifting} ]
First, $-L_1V^\top = V^\top(\hat{A}^l+\hat{A}^u)$ (\cite{Schaub2020RandomWO}, Lemma 3.2). Then 
observe that by definition, $\hat{\textbf{P}} = \hat{A}\text{diag}(\hat{A}\textbf{1})^{-1} =
\hat{A}\hat{W}^{-1} =  \frac{1}{2}\hat{A}
\begin{bmatrix}
D_{tot}^{-1} & 0 \\
0 & D_{tot}^{-1}
\end{bmatrix}$
and $\hat{A}^s =\begin{bmatrix}
D_{tot} & 0 \\
0 & D_{tot}
\end{bmatrix} $. So $V^\top D_{tot} = V^\top \hat{A}^s$. It follows that $$V^\top \hat{A}\hat{W}^{-1} = (D_{tot} - L_1)V^\top\hat{W}^{-1} = (D_{tot} - L_1)\frac{1}{2}D_{tot}^{-1}V^\top = \frac{1}{2}(I-\mathcal{L}_1)V^\top\,.$$
\end{proof}

\begin{proof}[Proof of Proposition~\ref{prop:block}]
The symmetry of $\hat{A}$ follows directly from the symmetry of $\hat{A}^l$ and $\hat{A}^u$. Writing out $\hat{A}^l$ and $\hat{A}^u$ explicitly using the formulas from Theorem~$\ref{thm:lifting}$ gives the block form.
\end{proof}

\begin{proof}[Proof of Theorem~\ref{thm:main}]
Consider the beginning state, where all the links are put in their own communities. Then according to \eqref{eq:mod}, the definition of modularity, putting oriented link $e_0$ into the community $\hat{C}$ will incur a change in modularity by 

$$\Delta Q(e_0 \to \hat{C}) = \frac{\hat{K}_{e_0,\hat{C}}}{2\hat{m}} - \frac{\sum_{tot}\hat{k}_{e_0}}{2\hat{m}^2}\,,$$
where 
\begin{itemize}
    \item $\sum_{tot} = \sum_{e\in \hat{C}}\hat{k}_e$  is sum of all link weights in $\hat{C}$;
    \item $\hat{K}_{e_0,\hat{C}} = \sum_{e\in \hat{C}}\hat{A}_{e_0,e} + \sum_{e\in \hat{C}}\hat{A}_{e,e_0}$ is the sum of link weights connecting $e_0$ and $\hat{C}$.
\end{itemize}
Then $\Delta Q(e_0 \to \bar{e}_0) \geq 0$, following from $(\ast)$. Moreover, for any other neighboring nodes $e'$ for $e_0$, $\Delta Q(e_0 \to \bar{e}_0) \geq \Delta Q(e_0 \to e')\,,$
which can be shown by considering
\begin{align*}
    \Delta Q(e_0 \to \bar{e}_0) - \Delta Q(e_0 \to e') = \frac{\hat{K}_{e_0,\bar{e}_0}-\hat{K}_{e_0,e'}}{2\hat{m}} - \frac{(\hat{k}_{\bar{e}_0}-\hat{k}_{e'})\hat{k}_{e_0}}{2\hat{m}^2}\,.
\end{align*}
Notice that $\hat{K}_{e_0,e}$ maximizes with $e = \bar{e}_0$ such that $\hat{K}_{e_0,\bar{e}_0} = 2(2 + deg(e_0))$, and the difference with other $\hat{K}_{e_0,e}$ is at least 2, while $(\hat{k}_{\bar{e}_0}-\hat{k}_{e'})\hat{k}_{e_0}\leq 2\hat{m}$. Hence  $\Delta Q(e_0 \to \bar{e}_0) - \Delta Q(e_0 \to e') \geq 0, \forall e'$. The Louvain method will put $e_0$ into the community of $\bar{e}_0$.

Then consider $e_1$. Let the community of $e_0$ and $\bar{e}_0$ be $\hat{C}_{0}$. Based on a similar reasoning as above, it suffices to show that the Louvain method will put $e_1$ in the community of $\bar{e}_1$, not $\hat{C}_0$. Notice that 
\begin{align*}
    \Delta Q(e_1 \to \bar{e}_1) - \Delta Q(e_0 \to \hat{C}_0) = \frac{\hat{K}_{e_1,\bar{e}_1}-\hat{K}_{e_1,\hat{C}_0}}{2\hat{m}} - \frac{(\hat{k}_{\bar{e}_1}-2\hat{k}_{e_0})k_{e_1}}{2\hat{m}^2}\,.
\end{align*}
We have  $\hat{K}_{e_1,\bar{e}_1} = 2(2 + deg(e_1))$, and $\hat{K}_{e_1,\hat{C}_0} \leq  2(1 + deg(e_1))$. So  $\Delta Q(e_1 \to \bar{e}_1) - \Delta Q(e_0 \to \hat{C}_0) \geq 0$, and we will put $e_1$ into the community of $\bar{e}_1$.
Apply the same reasoning to $e_2,...,\bar{e}_N$, after we iterate links with the chosen orientations, we get communities $\hat{C}_0 = \{e_0,\bar{e}_0\}, \hat{C}_1 = \{e_1,\bar{e}_1\},...,\hat{C}_N = \{e_N,\bar{e}_N\}.$

Finally, we notice that by the symmetry between $e$ and $\bar{e}$, $\hat{C}_0$, $\hat{C}_1$, ..., $\hat{C}_N$ is a stable partition of $\G$, in the sense that further iteration of moving a single node is not going to increase $Q$. Thus we conclude that we end the first phase of the Louvain algorithm with $\hat{C}_0$, $\hat{C}_1$, ..., $\hat{C}_N$, and the second phase of the algorithm starts with $N$ supernodes---where we collapse two orientations of the same link as one node.
\end{proof}

\begin{proof}[Proof of Theorem~\ref{cor:main}]
The result follows directly from Proposition~\ref{prop:block} and Theorem~\ref{thm:main}.
\end{proof}

\begin{proof}[Proof of Proposition~\ref{prop:pblock} ]
The results follows directly from Proposition~\ref{prop:block}, the block form of $\hat{A}$.
\end{proof}

\begin{proof}[Proof of Proposition~\ref{prop:decom} ]
\begin{align*}
&\det\begin{pmatrix}
\hat{P}_I - \lambda I & \hat{P}_{II} \\
\hat{P}_{II} & \hat{P}_{I} -\lambda I 
\end{pmatrix} \\
= \,\, &\det\left(\hat{P}_I +\hat{P}_{II} - \lambda I\right)\cdot \det\left(\hat{P}_I -\hat{P}_{II} - \lambda I\right)\,.
\end{align*}
\end{proof}

\begin{proof}[Proof of Proposition~\ref{prop:even}]
$$\begin{bmatrix}
\hat{P}_I & \hat{P}_{II} \\
\hat{P}_{II} & \hat{P}_{I}
\end{bmatrix}\begin{bmatrix}
x\\
x
\end{bmatrix} = 
\begin{bmatrix}
(\hat{P}_I +\hat{P}_{II})x \\
(\hat{P}_I +\hat{P}_{II})x
\end{bmatrix}
=\lambda\begin{bmatrix}
x\\
x
\end{bmatrix}\,.
$$
\end{proof}

\begin{proof}[Proof of Corollary~\ref{cor:even}]
$\hat{P}_I + \hat{P}_{II}$ is a stochastic matrix describing the random walk on the lifted ``supernode" graph $S(\G)$ with adjacency matrix $\hat{A}_I + \hat{A}_{II}$. So $1$ is an eigenvalue for $\hat{P}_I + \hat{P}_{II}$. The rest follows from Proposition~\ref{prop:even}.
\end{proof}

\begin{proof}[Proof of Corollary~\ref{cor:stationary}]
Notice that in Corollary~\ref{cor:even}, $x$ is the stationary solution to $\hat{P}_I + \hat{P}_{II}$. Let $\pi = x^\top$, then $\hat{\pi} = [\pi\,\, \pi]^\top$.
\end{proof}

\section{Data Descriptions}\label{app:data}
Here are detailed descriptions of the seven real-world datasets used in section~\ref{experiments}:
\begin{itemize}
    \item Physical contact data (\texttt{contact-primary-school}, \texttt{contact-high-school}): nodes are individuals, and simplices form when individuals are in proximity of one another within a short time period. 
    \paragraph{Community quality} For the metadata that measure node similarity, we used the classroom that each student belongs to. \textit{e.g.} $s(i,j) = 1$ if and only if two students belong to the same classroom. We made a similar hypothesis as Ahn et al. had in the mobile phone network that social contact is more frequent for people that are geographically related~\cite{Ahn2010LinkCR}. This phenomenon was indeed observed in the original contact data collection processes~\cite{Stehl2011HighResolutionMO,Mastrandrea2015ContactPI}.
    \paragraph{Overlap quality} For the metadata that can serve as a reasonable proxy for the number of communities that each node belongs to, we used the total number of contacts that each individual made during the observation window. Again, this operated under the similar assumption as Ahn et al. had in the mobile phone network that frequent contact makers may fulfill broader roles in their social networks~\cite{Ahn2010LinkCR}. 
    \newline
    \item Communication data (\texttt{email-Eu},  \texttt{email-Enron}): nodes are email addresses and a simplex is formed if individuals send one another emails over a time period. For \texttt{email-Eu}, the time period is the entire observation window; for \texttt{email-Enron}, the time period is seven days.
    \paragraph{Community quality} For the metadata that measure node similarity, we similarly used the department memberships (\texttt{email-Eu}) and job positions (\texttt{email-Enron}) of the nodes, assuming that people communicate more often if they belong to the same department or have similar job titles. 
    \paragraph{Overlap quality} For the metadata that can serve as a reasonable proxy for the number of communities each node participates in,  we likewise used the total number of emails that each email address sent over the observation window.
    \newline
    \item U.S. Congress collaboration data (\texttt{congress-bills}, \texttt{senate-committees}, \texttt{house-committees}): nodes are members of Congress and a simplex is the set of members co-sponsoring a bill or being part of a same committee.
    \paragraph{Community quality} For congress data, Ahn et al. measured the political and ideological similarity between each pair of congressmen~\cite{Ahn2010LinkCR}. we used political party affiliation as our similarly measure for \texttt{congress-bills}, \texttt{senate-committees} and \texttt{house-committees}. 
    \paragraph{Overlap quality} Ahn et al. used the seniority of each congressperson, measured as the number of elected terms that person has served under the assumption that longer-serving member would more easily participate in multiple collaborations than those who are more newly elected~\cite{Ahn2010LinkCR} . we adopted the number of elected terms as our proxy for \texttt{congress-bills}, \texttt{senate-committees} and \texttt{house-committees}. 
\end{itemize}

\end{document}